\documentclass[preprint,a4paper,11pt]{elsarticle}
\usepackage[T1]{fontenc}
\usepackage[utf8]{inputenc}
\usepackage[english]{babel}
\usepackage{amsmath,amsthm,enumerate}
\usepackage{amssymb}
\usepackage{amscd}
\usepackage{vmargin}
\usepackage{url}
\usepackage{stmaryrd}
\usepackage{tikz,pgf}
\usepackage{subfig}
\usepackage{cancel}
\usepackage[textsize=footnotesize]{todonotes}
\usetikzlibrary{positioning}
\usepackage{subfig}
\usetikzlibrary{shapes.geometric}
\usepackage{soul}
  
\begin{document}
\begin{frontmatter}

\title{Further Evidence Towards the Multiplicative 1-2-3 Conjecture}

\author[nice]{Julien Bensmail}
\author[labri]{Herv\'e Hocquard}
\author[labri]{Dimitri Lajou}
\author[labri]{\'Eric Sopena}

\address[nice]{Universit\'e C\^ote d'Azur, CNRS, Inria, I3S, France}
\address[labri]{Univ. Bordeaux, CNRS,  Bordeaux INP, LaBRI, UMR 5800, F-33400, Talence, France}

\journal{...}

\begin{abstract}
The product version of the 1-2-3 Conjecture, introduced by Skowronek-Kazi\'ow in 2012, states that, a few obvious exceptions apart, all graphs can be $3$-edge-labelled so that no two adjacent vertices get incident to the same product of labels. To date, this conjecture was mainly verified for complete graphs and $3$-colourable graphs. As a strong support to the conjecture, it was also proved that all graphs admit such $4$-labellings.

In this work, we investigate how a recent proof of the multiset version of the 1-2-3 Conjecture by Vu\v{c}kovi\'c can be adapted to prove results on the product version. We prove that $4$-chromatic graphs verify the product version of the 1-2-3 Conjecture. We also prove that for all graphs we can design $3$-labellings that almost have the desired property. This leads to a new problem, that we solve for some graph classes.
\end{abstract}

\begin{keyword} 
1-2-3 Conjecture; multiset version; product version; $4$-chromatic graphs.
\end{keyword}
 
\end{frontmatter}

\newtheorem{theorem}{Theorem}[section]
\newtheorem{lemma}[theorem]{Lemma}
\newtheorem{conjecture}[theorem]{Conjecture}
\newtheorem{observation}[theorem]{Observation}
\newtheorem{claim}[theorem]{Claim}
\newtheorem{corollary}[theorem]{Corollary}
\newtheorem{proposition}[theorem]{Proposition}
\newtheorem{question}[theorem]{Question}
\newtheorem*{123c}{1-2-3 Conjecture (sum version)}
\newtheorem*{m123c}{1-2-3 Conjecture (multiset version)}
\newtheorem*{p123c}{1-2-3 Conjecture (product version)}

\newcommand{\qedclaim}{\hfill $\diamond$ \medskip}
\newenvironment{proofclaim}{\noindent{\em Proof of the claim.}}{\qedclaim}

\newcommand{\chis}{\chi_{\rm S}}
\newcommand{\chim}{\chi_{\rm M}}
\newcommand{\chip}{\chi_{\rm P}}


\section{Introduction}

This work takes place in the general context of \textbf{distinguishing labellings}, where the aim, given an undirected graph, is to label its edges so that its adjacent vertices get distinguished by some function computed from the labelling. Formally, a \textit{$k$-labelling} $\ell: E(G) \rightarrow \{1,\dots,k\}$ of a graph $G$ assigns a label from $\{1,\dots,k\}$ to each edge, and, for every vertex $v$, we can compute some function $f(v)$ of the labels assigned to the edges incident to $v$. The goal is then to design $\ell$ so that $f(u) \neq f(v)$ for every edge $uv$ of $G$.
As reported in a survey~\cite{Gal98} by Gallian on the topic, there actually exist dozens and dozens types of distinguishing labelling notions, which all have their own particular behaviours and subtleties.

\medskip

We are here more particularly interested in the so-called \textbf{1-2-3 Conjecture}, which is defined through the following notions. Given a labelling $\ell$ of a graph $G$, we can compute for every vertex $v$ its sum $\sigma_\ell(v)$ of incident labels, being formally $\sigma_\ell(v)=\Sigma_{u \in N(v)} \ell(uv)$. We say that $\ell$ is \textit{s-proper} if the so-obtained $\sigma_\ell$ yields a proper vertex-colouring of $G$, i.e., $\sigma_\ell(u) \neq \sigma_\ell(v)$ for every edge $uv$. Generally speaking, not only we aim at finding s-proper $k$-labellings of $G$, but also we aim at designing such ones having $k$ as small as possible. Thus, for $G$, we are interested in determining $\chis(G)$, which is the smallest $k \geq 1$ such that s-proper $k$-labellings of $G$ do exist.

Greedy arguments show that there exists only one connected graph $G$ for which $\chis(G)$ is not defined, and that graph is $K_2$. This implies that $\chis(G)$ is defined for every graph $G$ with no component isomorphic to $K_2$, which we call a \textit{nice graph}. It is then legitimate to wonder how large can $\chis(G)$ be in general, for a nice graph $G$. Karo\'nski, {\L}uczak and Thomason conjectured that this value cannot exceed~$3$ in general~\cite{KLT04}:

\begin{123c}
If $G$ is a nice graph, then $\chis(G) \leq 3$.
\end{123c}

One could naturally wonder about slight modifications of the 1-2-3 Conjecture, where the aim would be to design labellings $\ell$ distinguishing adjacent vertices accordingly to a function $f$ that is somewhat close to the sum function $\sigma_\ell$. There actually exist at least two such variants, to be described in what follows, which sound particularly interesting due to their respective subtleties, to some behaviours they share with the original 1-2-3 Conjecture, and to general existing connections with that conjecture.

\begin{itemize}
    \item The first such variant we consider is the one where adjacent vertices of a graph $G$ are required, by a labelling $\ell$, to be distinguished by their \textit{multisets} of incident labels. Recall that a multiset is a set in which elements can be repeated. For a vertex $v$ of $G$, we denote by $\mu_\ell(v)$ the multiset of labels assigned to the edges incident to $v$. We say that $\ell$ is \textit{m-proper} if $\mu_\ell$ is a proper vertex-colouring of $G$, while we denote by $\chim(G)$ the least $k \geq 1$ such that $G$ admits m-proper $k$-labellings (if any).
    
    \item The second such variant is the one where adjacent vertices of $G$ must be, by $\ell$, distinguished accordingly to the \textit{products} of their incident labels. Formally, for a vertex $v$ of $G$, we define $\rho_\ell(v)$ as the product of labels assigned to the edges incident to $v$. We say that $\ell$ is \textit{p-proper} if $\rho_\ell$ is a proper vertex-colouring of $G$. We denote by $\chip(G)$ the smallest $k \geq 1$ such that $G$ admits p-proper $k$-labellings (if any).
\end{itemize}

There exist several interesting connections between the previous three series of notions. For instance, it can be easily noted that an s-proper or p-proper labelling is always m-proper. As a result, $\chim(G) \leq \min \{\chis(G),\chip(G)\}$ holds for every graph $G$ for which the parameters are defined (see below). In general, there is no other systematic relationship between these three notions, though some exist in particular contexts. For instance, s-proper $2$-labellings, m-proper $2$-labellings and p-proper $2$-labellings are equivalent notions in regular graphs~\cite{BBDHPSW19}. It can also be noted that s-proper $\{0,1\}$-labellings and p-proper $\{1,2\}$-labellings are equivalent notions~\cite{Lyn18}. Another illustration is that an m-proper $k$-labelling yields a p-proper $\{l_1,\dots,l_k\}$-labelling, for any set $\{l_1,\dots,l_k\}$ of $k$ pairwise coprime integers.

Just as for the 1-2-3 Conjecture, one can wonder how large can $\chim(G)$ and $\chip(G)$ be for a given graph $G$. Before providing hints on that very question, let us first mention that, similarly as for s-proper labellings, the only connected graph admitting no m-proper labellings and no p-proper labellings is $K_2$. Thus, the notion of nice graph coincides for the three types of proper labellings. It actually turns out that the straight analogue of the 1-2-3 Conjecture is believed to hold for m-proper labellings and p-proper labellings; namely:

\begin{m123c}
If $G$ is a nice graph, then $\chim(G) \leq 3$.
\end{m123c}

\begin{p123c}
If $G$ is a nice graph, then $\chip(G) \leq 3$.
\end{p123c}

The multiset version of the 1-2-3 Conjecture was introduced by Addario-Berry, Aldred,  Dalal and Reed in~\cite{AADR05}, while the product version was introduced by Skowronek-Kazi\'ow in~\cite{SK12}.
By an argument above, recall that the sum version and the product version of the 1-2-3 Conjecture, if true, would actually imply the multiset version. From that angle, the multiset version does appear, at least intuitively, as the most feasible out of the three versions. This is reinforced by unique behaviours of m-proper labellings over s-proper labellings and p-proper labellings. In particular, by a labelling $\ell$ of a graph, in order to have $\mu_\ell(u)=\mu_\ell(v)$ for any two vertices $u$ and $v$, note that we must have $d(u)=d(v)$.

The most notable facts towards these three variants of the 1-2-3 Conjecture to date are:

\begin{itemize}
    \item Regarding the sum version, the best result to date, proved by Kalkowski, Karo\'nski and Pfender in~\cite{KKP10}, is that $\chis(G) \leq 5$ holds for every nice graph $G$. The conjecture was verified for all $3$-colourable graphs~\cite{KLT04}. Regarding $4$-chromatic graphs, the conjecture was verified for $4$-edge-connected ones~\cite{WZZ17}. In~\cite{Prz19c}, it was recently shown that $\chis(G) \leq 4$ holds for every nice regular graph $G$.
    
    \item Regarding the multiset version, for long the best result, proved by Addario-Berry, Aldred,  Dalal and Reed in~\cite{AADR05}, was that $\chim(G) \leq 4$ holds for every nice graph $G$. A few years ago, a breakthrough result was obtained by Vu\v{c}kovi\'c in~\cite{Vuc18}, in which he gave a full proof of the conjecture.
    \end{itemize}
    
    \begin{theorem}[\cite{Vuc18}]\label{theorem:vuckovic}
    If $G$ is a nice graph, then $\chim(G) \leq 3$.
    \end{theorem}
    
    \begin{itemize}
    \item Regarding the product version, the best results to date were mainly obtained via adaptations of arguments used to provide results towards the sum and multiset versions. Specifically, Skowronek-Kazi\'ow proved in~\cite{SK12} that $\chip(G) \leq 4$ holds for all nice graphs $G$. In the same article, she proved the product version of the 1-2-3 Conjecture for $3$-colourable graphs.
\end{itemize}

In this work, we provide several results towards the product version of the 1-2-3 Conjecture. A first (minor) reason for focusing on this version is that it is, out of the three versions, the least investigated one to date. A second (major) reason stems from the recent proof of Theorem~\ref{theorem:vuckovic} by Vu\v{c}kovi\'c. As pointed out earlier, m-proper labellings and p-proper labellings tend to have alike behaviours, which gives us  hope that the proof of Vu\v{c}kovi\'c might be a step towards proving the product version of the problem.

Let us support this perspective further. 
It is first important to mention that the labels $1,2,3$ by a $3$-labelling  are very special in terms of vertex products. Note in particular that label~$1$ has a  unique behaviour, since assigning label~$1$ to an edge $uv$ by a labelling $\ell$ impacts neither $\rho_\ell(u)$ nor $\rho_\ell(v)$. It is important, however, to emphasise that assigning label~$1$ to $uv$ is not similar to deleting $uv$ from the graph, as, though $\ell(uv)$ does not contribute to $\rho_\ell(u)$ and $\rho_\ell(v)$, it requires $\rho_\ell(u)$ and $\rho_\ell(v)$ to be different by a p-proper $3$-labelling $\ell$. Because $2$ and~$3$ are coprime, this implies that, in order for $\rho_\ell(u) \neq \rho_\ell(v)$ to hold, the decomposition of $\rho_\ell(u)$ into prime factors must differ from that of $\rho_\ell(v)$. In other words, if we denote by $d_i(w)$ the \textit{$i$-degree} of a vertex $w$ by a labelling as the number of edges incident to $w$ assigned label~$i$, then, by a $3$-labelling $\ell$, $\rho_\ell(u) \neq \rho_\ell(v)$ holds if and only if either $d_2(u) \neq d_2(v)$ or $d_3(u) \neq d_3(v)$.

That last property makes labels~$2$ and~$3$ by a $3$-labelling $\ell$ very close in terms of vertex multisets and vertex products, since also $\mu_\ell(u) \neq \mu_\ell(v)$ holds as soon as $d_2(u) \neq d_2(v)$ or $d_3(u) \neq d_3(v)$. Thus, the difference between m-proper $3$-labellings and p-proper $3$-labellings only lies in the behaviour of label~$1$: for the first objects, every edge $uv$ labelled~$1$ contributes to both $\mu_\ell(u)$ and $\mu_\ell(v)$, while, for the second objects, every edge $uv$ labelled~$1$ contributes to none of $\rho_\ell(u)$ and $\rho_\ell(v)$. For that reason, m-proper $3$-labellings are not p-proper in general; however, there are contexts where this is the case, such as the following meaningful one:

\begin{observation}
Nice regular graphs verify the product version of the 1-2-3 Conjecture.
\end{observation}

\begin{proof}
Let $G$ be a nice $\Delta$-regular graph.
By Theorem~\ref{theorem:vuckovic}, there exists an m-proper $3$-labelling $\ell$ of $G$. We claim it is also p-proper. Indeed, by arguments above, if $\rho_\ell(u)=\rho_\ell(v)$ holds for some edge $uv$, then $d_2(u)=d_2(v)$ and $d_3(u)=d_3(v)$. Since $d(u)=d(v)=\Delta$, this means also $d_1(u)=d_1(v)$ holds. We then deduce that $\mu_\ell(u)=\mu_\ell(v)$ holds, a contradiction. Thus, no two adjacent vertices of $G$ have the same product of labels.
\end{proof}

Our main intention in this paper is to investigate how the mechanisms in the proof of Theorem~\ref{theorem:vuckovic} can be used in the product setting. In Section~\ref{section:4chromatic}, we first prove that the product version of the 1-2-3 Conjecture holds for $4$-chromatic graphs, which goes beyond the best known such result to date for the sum version, which is, as stated earlier, that $4$-edge-connected $4$-chromatic graphs verify the sum version of the 1-2-3 Conjecture~\cite{WZZ17}. In Section~\ref{section:low-conflicts}, we give a result that is close to the product version of the conjecture, as we describe how to design, for any nice graph, $3$-labellings that are very close to be p-proper. This leads us to raising a conjecture on almost p-proper $2$-labellings in Section~\ref{section:weak-conjecture}, that matches an existing weakening of the sum version of the 1-2-3 Conjecture from~\cite{GWW15}. We finally verify our conjecture for several classes of graphs.


\section{The Multiplicative 1-2-3 Conjecture for $4$-chromatic graphs}\label{section:4chromatic}

For a graph $G$, a \textit{proper $k$-vertex-colouring} is a partition $(V_1, \dots, V_k)$ of $V(G)$ into independent sets, and the \textit{chromatic number} $\chi(G)$ of $G$ is the smallest $k$ such that there exist proper $k$-vertex-colourings of $G$. Recall that $G$ is \textit{$k$-chromatic} if its chromatic number is exactly~$k$. Equivalently, $G$ is $k$-chromatic if it admits proper $k$-vertex-colourings, but no proper $k'$-vertex-colourings with $k'<k$.

In this section, we prove the following:

\begin{theorem}\label{theorem:4chromatic}
If $G$ is a $4$-chromatic graph, then $\chip(G) \leq 3$.
\end{theorem}

\begin{proof}
Let $(V_1,V_2,V_3,V_4)$ be a proper $4$-vertex-colouring of $G$. For any vertex $v \in V_i$, an \textit{upward edge} (resp. \textit{downward edge}) is an incident edge going to a vertex in some $V_j$ with $j<i$ (resp. $j>i$). Note that all vertices in $V_1$ have no upward edges, while all vertices in $V_4$ have no downward edges. 

Free to move vertices from part to part, we may assume that, for every vertex $v$ in any $V_i$ with $i \in \{2,3,4\}$, there is an upward edge going to each of $V_1,\dots,V_{i-1}$. Indeed, if there is a $j<i$ such that $v$ has no upward edge to $V_j$, then by moving $v$ to $V_j$ we obtain another $4$-vertex-colouring of $G$ that is proper. By repeating this moving process as long as needed, we eventually reach a proper $4$-vertex-colouring with the desired property. In particular, note that the process finishes since vertices are only moved to parts with lower index. Furthermore, since $G$ is $4$-chromatic, none of the four parts can become empty.

A p-proper $3$-labelling $\ell$ of $G$ will be obtained through two main stages. The first stage will consist in considering the vertices in $V_4$ and $V_3$, and labelling their upward edges so that particular types of products are obtained for their vertices, guaranteeing that none of them are in conflict (i.e., have the same product of incident labels). In particular, these vertices will be \textit{bichromatic}, meaning that they have both $2$-degree and $3$-degree at least~$1$. On the contrary, in general, a few cases apart, the vertices in $V_2$ and $V_1$ will be \textit{monochromatic}, meaning that they are not bichromatic. More precisely, for $i \in \{1,2,3\}$, a vertex $v$ is \textit{$i$-monochromatic} if it is incident only to edges labelled~$i$ or~$1$ and $d_i(v) > 0$. In other words, a $1$-monochromatic vertex is a vertex $v$ with $\rho_\ell(v)=1$, while an $i$-monochromatic vertex with $i \in \{2,3\}$ is a vertex $v$ with $\rho_\ell(v)=i^x$ for some $x \geq 1$. Right after the first stage, in a second stage, we will label the edges joining vertices in $V_2$ and $V_1$ to get rid of all remaining conflicts.

As mentioned in the introductory section, recall that two vertices can only be in conflict if they have the same $2$-degree and the same $3$-degree. This means that an $i$-monochromatic vertex and a $j$-monochromatic vertex can only be in conflict if $i=j$, and that a monochromatic vertex and a bichromatic vertex can never be in conflict. More generally speaking, a vertex with $2$-degree~$i$ and $3$-degree~$j$ and a vertex with $2$-degree~$i'$ and $3$-degree~$j'$ can only be in conflict if $i=i'$ and $j=j'$.

\medskip

To ease the understanding of the proof, we start from $\ell$ being the labelling of $G$ assigning label~$1$ to all edges. We then modify some labels so that particular products are obtained for some vertices. We describe this modification process step by step, so that, after each of these steps, we can point out the consequences of our modifications. In particular, the reader should keep in mind that, at any point, any existing conflict is intended to be dealt with in later stages of the modification process. In, particular, note that, at the beginning, all vertices are $1$-monochromatic and are thus all in conflict.

\paragraph{Step 1: Labelling the upward edges of $V_4$ and $V_3$} \mbox{}

\medskip

We first relabel all upward edges of the vertices in $V_4$. To that end, we consider every $v \in V_4$ in arbitrary order, and apply the following:

\begin{enumerate}
    \item for every upward edge $vu$ with $u \in V_1$, we set $\ell(vu)=2$;
    
    \item for every upward edge $vu$ with $u \in V_3$, we set $\ell(vu)=3$;
    
    \item if currently $v$ has odd $3$-degree, then we pick an arbitrary upward edge $vu$ with $u \in V_2$, and set $\ell(vu)=3$.
\end{enumerate}

Recall that such upward edges exist by our original assumption on $(V_1,V_2,V_3,V_4)$.
Also, note that, at this point, the vertices in $V_4$ verify the following:

\begin{claim}\label{claim:V4}
Every vertex of $V_4$ is bichromatic with even $3$-degree.
\end{claim}

Furthermore, at this point, every vertex $v$ in $V_3$ has all its downward edges (if any) assigned label~$3$ by $\ell$. Now, for every such $v \in V_3$, we modify the label of the upward edges as follows:

\begin{enumerate}
    \item for every upward edge $vu$ with $u \in V_1$, we set $\ell(vu)=2$;
    
    \item if currently $v$ has even $3$-degree, then we pick an arbitrary upward edge $vu$ with $u \in V_2$, and set $\ell(vu)=3$.
\end{enumerate}

Note that Claim~\ref{claim:V4} is not impacted by these modifications. Furthermore, it can be checked that the vertices in $V_3$ fulfil the following:

\begin{claim}\label{claim:V3}
Every vertex $v$ of $V_3$ is bichromatic with odd $3$-degree.
\end{claim}

\begin{proofclaim}
If $v$ has downward edges to $V_4$, then they are labelled~$3$ in which case $v$ is bichromatic (regardless of whether the second item applies or not). If $v$ has no downward edges, then the second item of the process applies, and $v$ gets bichromatic by labelling~$3$ an upward edge to $V_2$. In both cases, $v$ has its $3$-degree being of the desired parity.
\end{proofclaim}

Note that Claims~\ref{claim:V4} and~\ref{claim:V3} imply that any two adjacent vertices in $V_4$ and $V_3$ cannot be in conflict, due to their different $3$-degrees. Furthermore, it can be checked that, at this point, the vertices in $V_2$ and $V_1$ meet the following properties:

\begin{claim}\label{claim:V2afterV3}
For every vertex of $V_2$:
\begin{itemize}
    \item all downward edges to $V_4$ and $V_3$ are labelled~$1$ or~$3$;
    
    \item all upward edges to $V_1$ are labelled~$1$.
\end{itemize}
\end{claim}

\begin{claim}\label{claim:V1afterV3}
For every vertex of $V_1$:
\begin{itemize}
    \item all downward edges to $V_4$ and $V_3$ are labelled~$2$;
    
    \item all downward edges to $V_2$ are labelled~$1$.
\end{itemize}
\end{claim}

All previous claims imply that, at the moment, only vertices in $V_2$ and $V_1$ can be in conflict. More precisely, every vertex of $V_1$ is currently either $1$-monochromatic or $2$-monochromatic, while every vertex of $V_2$ is either $1$-monochromatic or $3$-monochromatic. Thus, two adjacent vertices in $V_2$ and $V_1$ can only be in conflict if they are both $1$-monochromatic.
The next stage is dedicated to getting rid of these conflicts.

\paragraph{Step 2: Labelling the edges between $V_1$ and $V_2$}\mbox{}

\medskip

For every vertex $v$ of $G$, we define its \textit{$\{2,3\}$-degree} as the sum $d_2(v)+d_3(v)$ of its $2$-degree and $3$-degree.
It is important to mention that, in all modifications we apply to $\ell$ from this point on, the only way for the $\{2,3\}$-degree of a vertex $v$ in $V_3 \cup V_4$ to change is via setting to~$2$ the label of edges $uv$ with $u \in V_2 \cup V_1$. In particular, note that Claims~\ref{claim:V4} and~\ref{claim:V3} are not impacted by such modifications, as they only alter  $2$-degrees. Hence, through performing such modifications, adjacent vertices in $V_4$ and $V_3$ cannot get in conflict.

For the whole step, we define $\mathcal{H}$ as the set of (connected) components induced by the edges joining $1$-monochromatic vertices of $V_2$ and vertices of $V_1$ (of any type).
Note that any two conflicting vertices at this point, i.e., adjacent vertices being currently $1$-monochromatic, are part of a component $H$ of $\mathcal H$. In what follows, we call an $H$ containing such a pair of conflicting vertices a \textit{conflicting component}. Our main goal here is now to apply local label modifications to get rid of all conflicting components of $\mathcal H$.

It is important to note that two vertices from two distinct components $H_1$ and $H_2$ of $\mathcal{H}$ cannot be adjacent. Assume indeed that $uv$ is an edge of $G$, where $v \in V(H_1) \cap V_2$ and $u \in V(H_2) \cap V_1$. By definition of $\mathcal H$, this means that $v$ is $1$-monochromatic, in which case $H_1$ and $H_2$ altogether induce a component of $\mathcal{H}$. A consequence is that we can freely treat the conflicting components of $\mathcal{H}$ independently.

\medskip

Consider a conflicting component $H \in \mathcal H$.
Note that we would be done with $H$ if we could get rid of all conflicts in $H$ by relabelling its edges so that all vertices in $H$ remain $1$-monochromatic or $2$-monochromatic, as, this way, no conflict with vertices in $V_3$ or $V_4$ would arise (by Claims~\ref{claim:V4} and~\ref{claim:V3}). This is a configuration that can actually almost be attained, in the following sense:

\begin{claim}\label{claim:bipartite}
For every vertex $v$ in any part $V_i \in \{V_1,V_2\}$ of $H$, we can relabel the edges of $H$ with $1$ and $2$ so that $d_2(u)$ is odd for every $u \in V_i \setminus \{v\}$, and $d_2(u)$ is even for every $u \in V_{3-i}$. Similarly, we can relabel the edges of $H$ with $1$ and $2$ so that $d_2(u)$ is even for every $u \in V_i \setminus \{v\}$, and $d_2(u)$ is odd for every $u \in V_{3-i}$.
\end{claim}

\begin{proofclaim}
Assume the conditions of the statement are not already met.
So let us consider any vertex $u$ different from $v$ that does not verify the desired condition.
Since $H$ is connected, there is a path $P$ from $u$ to $v$ in $H$. Now traverse $P$ from $u$ to $v$, and, as going along, switch the label of every traversed edge from~$1$ to~$2$ and \textit{vice versa}. Note that this switching procedure has the following effects:

\begin{itemize}
    \item for every internal vertex of $P$, the parity of its $2$-degree is not altered;
    
    \item for each of the two ends $u$ and $v$ of $P$, the parity of its $2$-degree is altered.
\end{itemize}

This way, note that $u$ now satisfies its desired condition, while, for all vertices of $H$ different from $u$ and $v$, the situation regarding their desired condition has not changed.

By repeating this switching procedure as long as desired, we eventually get that all vertices different from $v$ have their $2$-degree meeting the desired parity condition.
\end{proofclaim}

To deal with $H$, we now apply certain label modifications depending on the surroundings of $H$. We start off by considering the following three cases. In each case, it is implicitly assumed that the previous ones do not apply.

\begin{itemize}
    \item \textbf{Case~1.} There is a vertex $v \in V(H) \cap V_2$ with a neighbour $w \in V_3 \cup V_4$. 
    
    Recall that $v$ is $1$-monochromatic (by definition of $\mathcal H$), and thus $vw$ is currently labelled~$1$. According to Claim~\ref{claim:bipartite}, in $H$ we can relabel edges with~$1$ to~$2$ so that all vertices in $V(H) \cap V_1$ have even $2$-degree while all vertices in $V(H) \cap V_2 \setminus \{v\}$ have odd $2$-degree. If also $v$ has odd $2$-degree, then no conflict remains in $H$. Otherwise, i.e., $v$ has even $2$-degree, then we change the label of $vw$ to $2$. As a result, $v$ now gets odd $2$-degree as well, while $w$ remains bichromatic with the same $3$-degree. Thus, no conflict remains in $H$, and no new conflict is created in $G$.
    
    \item \textbf{Case~2.} There is a $1$-monochromatic vertex $u \in V(H) \cap V_1$ with a $3$-monochromatic neighbour $v \in V_2$.
    
    Since $u$ is $1$-monochromatic, by Claim~\ref{claim:V1afterV3} it has no neighbour in $V_3 \cup V_4$. Also, $uv$ is currently labelled~$1$. As in the previous case, according to Claim~\ref{claim:bipartite} we can relabel with~$1$ and~$2$ the edges of $H$ to reach a situation where all vertices in $V(H) \cap V_2$ have odd $2$-degree while all vertices in $V(H) \cap V_1 \setminus \{u\}$ have even $2$-degree. If $u$ also has even $2$-degree, then we are done. Otherwise, $u$ has odd $2$-degree (thus at least~$1$). In that case, we assign label~$3$ to $uv$. As a result, $u$ gets bichromatic with no such neighbour while $v$ remains $3$-monochromatic. Thus, no conflict remains in $H$, and no new conflict was created in $G$.
    
    \item \textbf{Case~3.} There is a $1$-monochromatic vertex $u \in V(H) \cap V_1$ with $p \geq 2$ neighbours $v_1,\dots,v_p \in V(H) \cap V_2$.
    
    Because $u$ is $1$-monochromatic, by Claim~\ref{claim:V1afterV3} it has no neighbour in $V_3 \cup V_4$. Also,
    since previous Cases~1 and~2 did not apply, $u$ has no $3$-monochromatic neighbour in $V_2$, thus all $v_i$'s are $1$-monochromatic, and the $v_i$'s have no neighbours in $V_3 \cup V_4$. We here consider $H'=H-u$. Let us denote by $C_1,\dots,C_r$ the components of $H'$. By Claim~\ref{claim:bipartite}, in each $C_j$ ($1 \leq j \leq r$) we can relabel the edges with~$1$ and~$2$ so that all vertices in $V(C_j) \cap V_1$ have even $2$-degree while all vertices in $V(C_j) \cap V_2$ but maybe one of the $v_i$'s have odd $2$-degree. Note that this can be attained since each of the $C_j$'s contains at least one of the $v_i$'s. Finally, assign label~$3$ to all $uv_i$'s. As a result, $u$ becomes $3$-monochromatic with $3$-degree at least~$2$, while its only $3$-monochromatic neighbours are possibly some of the $v_i$'s, in which case these have $3$-degree~$1$. In $C_j$, the only vertex (one of the $v_i$'s) that was possibly in conflict with some vertices has turned bichromatic or $3$-monochromatic, while its neighbours remain $1$-monochromatic or $2$-monochromatic. Recall also that the $v_i$'s that became bichromatic cannot be adjacent to another bichromatic vertex, as, in particular, they have no neighbours in $V_3 \cup V_4$. Thus no conflict remains in $H$, and no new conflict was created in $G$.
\end{itemize}

Consider any remaining conflicting component $H \in \mathcal H$. Because Cases~1 to~3 above did not apply to $H$, the following holds:

\begin{claim}\label{claim:remaining-compo-H}
For every remaining conflicting component $H \in \mathcal H$:
\begin{itemize}
    \item all $1$-monochromatic vertices $u \in V(H) \cap V_1$ have degree $1$ in $G$;
    \item all $1$-monochromatic vertices of $H$ have no bichromatic neighbours in $G$.
\end{itemize}
\end{claim}

\begin{proofclaim}
By definition of $\mathcal H$, all vertices of $V(H) \cap V_2$ are $1$-monochromatic, while all $1$-monochromatic vertices of $V(H) \cap V_1$ have no neighbours in $V_3 \cup V_4$ (Claim~\ref{claim:V1afterV3}). Since $H$ did not verify Case~1 above, it has no vertex in $V(H) \cap V_2$ with a neighbour in $V_3 \cup V_4$. Similarly, because Cases~2 and~3 above did not apply, $H$ has no $1$-monochromatic vertex in $V(H) \cap V_1$ having a $3$-monochromatic neighbour in $V_2$ or two $1$-monochromatic neighbours in $V(H) \cap V_2$. Then the claim holds.
\end{proofclaim}

Let us now repeatedly apply the following procedure to $H$:

\begin{itemize}
    \item As long as $H$ has a $1$-monochromatic vertex $v \in V(H) \cap V_2$ with two $1$-monochromatic neighbours $u_1,u_2 \in V(H) \cap V_1$, we assign label~$3$ to $vu_1$ and $vu_2$. 
\end{itemize}

Note that this raises no conflict. On the one hand, $u_1$ and $u_2$ get $3$-monochromatic with $3$-degree~$1$ while $v$ is their unique $3$-monochromatic neighbour, and it has $3$-degree~$2$. Recall that $v$ is actually the unique $3$-monochromatic neighbour of $u_1$ and $u_2$, since $u_1$ and $u_2$ have degree~$1$ in the whole of $G$. Conversely, $u_1$ and $u_2$ are the only $3$-monochromatic neighbours of $v$, since the only $3$-monochromatic vertices we create in $V(H) \cap V_1$ during this procedure have degree~$1$ in $G$, and thus in $H$.
Thus, no new conflict arises.

Once the previous procedure has been repeated as long as possible for every remaining conflicting component $H$ of $\mathcal H$, note that the remaining conflicts involve disjoint edges $uv$ where $u \in V_1$, $v \in V_2$, and $u$ and $v$ are $1$-monochromatic.
Furthermore, since previous Cases~1 to~3 did not apply to $H$, we deduce that $u$ has degree precisely~$1$ in $G$, that $v$ has no other $1$-monochromatic neighbour in $H$, and that $v$ has no neighbour in $V_3 \cup V_4$. Since $G$ is nice, we must have $d(v) \geq 2$, which means that $v$ has other neighbours in $V(H) \cap V_1$. Any such neighbour $u' \in V(H) \cap V_1$ must be $2$-monochromatic. Indeed, on the one hand, if $u'$ is $1$-monochromatic, then the repeated process above could have been applied once more to $H$. On the other hand, note that $u'$ cannot be $3$-monochromatic with $vu'$ being labelled~$1$, as, in the process above, only degree-$1$ vertices of $G$, and thus of $V(H) \cap V_1$, get $3$-monochromatic. Now, we assign label~$2$ to $u'v$ and label~$3$ to $vu$. As a result, $u'$ remains $2$-monochromatic, $v$ gets bichromatic, while $u$ gets $3$-monochromatic. Then no new conflict arises, because, in particular, $u'$ remains $2$-monochromatic with no such neighbours, $v$ cannot have bichromatic neighbours by Claim~\ref{claim:remaining-compo-H}, and $u$ is only neighbouring $v$.

Eventually, the $3$-labelling $\ell$ has no remaining conflicts, and is thus p-proper.
\end{proof}

\section{Restricted product conflicts by $3$-labellings}\label{section:low-conflicts}

In this section, we show how the proof of Theorem~\ref{theorem:4chromatic} can be generalized, to prove that all graphs admit $3$-labellings $\ell$ that are ``almost'' p-proper. By that, we mean that if there are product conflicts by $\ell$, then the structures induced by the conflicting vertices are somewhat weak.
This is with respect to the following notion.
Let $\ell$ be a labelling of a graph $G$. 
For any $x \geq 1$, we denote by $S_x$ the set of vertices $v$ of $G$ with $\rho_\ell(v)=x$. Rephrased differently, the product version of the 1-2-3 Conjecture states that every nice graph $G$ admits a $3$-labelling such that $S_x$ is an independent set for every $x \geq 1$.

In the next result, we prove that every graph admits a $3$-labelling where all $S_x$'s are independent, with the exception of perhaps $S_1$, 
which might induce independent edges.

\begin{theorem}\label{theorem:star}
Every graph $G$ admits a $3$-labelling such that $S_1$ induces a (possibly empty) matching and isolated vertices while all other $S_x$'s are independent sets.
\end{theorem}

\begin{proof}
We may assume that $G$ is connected.
If $G$ is $K_2$, then it suffices to assign label~$1$ to its only edge.
If $G$ is $3$-colourable, then $G$ admits a p-proper $3$-labelling (according to~\cite{SK12}), by which every $S_x$ is an independent set. The same conclusion holds if $G$ is $4$-chromatic, by Theorem~\ref{theorem:4chromatic}.
Thus, we may suppose that $G$ is $k$-chromatic for some $k \geq 5$. Let us thus consider $(V_1,\dots,V_k)$ a proper $k$-vertex-colouring of $G$, where $k = \chi(G)$. By similar arguments as in the proof of Theorem~\ref{theorem:4chromatic}, we may assume that every vertex $v \in V_i$ with $i>1$ has upward edges to every part $V_1,\dots,V_{i-1}$.

Just as in the proof of Theorem~\ref{theorem:4chromatic}, we start from a labelling $\ell$ of $G$ assigning label~$1$ to all edges. We then consider the vertices of $V_k, V_{k-1}, \dots, V_3$ following that order and modify the labels of their upward edges so that certain products are obtained. Eventually, we will handle the edges joining the vertices in $V_1$ and $V_2$ so that additional conditions are met to make sure that only particular conflicts remain.

During a first modification phase, we aim at having the vertices verifying the following:

\begin{itemize}
    \item $v \in V_1$: $1$-monochromatic or $2$-monochromatic;
    \item $v \in V_2$: $1$-monochromatic or $3$-monochromatic;
    
    \item $v \in V_3$: bichromatic, $2$-degree~$1$, and even $\{2,3\}$-degree;
    
    \item $v \in V_4$: bichromatic, $3$-degree~$2$, and odd $\{2,3\}$-degree;
    \item $v \in V_5$: bichromatic, $2$-degree~$2$, and even $\{2,3\}$-degree;
    
    \item ...
    
    \item $v \in V_{2n}$, $n \geq 3$: bichromatic, $3$-degree~$n$, and odd $\{2,3\}$-degree;
    \item $v \in V_{2n+1}$, $n \geq 3$: bichromatic, $2$-degree~$n$, and even $\{2,3\}$-degree;
    
    \item ...
\end{itemize}

We note that if we can produce a $3$-labelling with the vertex properties above,
then the only possible conflicts would be along edges $uv$ such that $u \in V_1$, $v \in V_2$, and both $u$ and $v$ are $1$-monochromatic. 
Indeed, two vertices $u$ and $v$ such that $u \in V_1 \cup V_2$ and $v \in V_3 \cup \dots \cup V_k$ cannot be in conflict since monochromatic vertices and bichromatic vertices cannot be in conflict. Two vertices $u$ and $v$ with $u \in V_{2n}$ and $v \in V_{2n'+1}$ for $n \geq 2$ and $n' \geq 1$ cannot be in conflict since vertices with different $\{2,3\}$-degrees cannot be in conflict. Finally, two vertices $u$ and $v$ with $u \in V_{2n+p}$ and $v \in V_{2n'+p}$ for $n \geq 2$, $n' \geq 3$ ($n \neq n'$) and $p \in \{0,1\}$ cannot be in conflict since bichromatic vertices can only be in conflict if they have the same $2$-degree and $3$-degree.

\medskip

Let us now describe how to modify $\ell$ so that the conditions above are met. We consider the vertices of $V_k,\dots,V_3$ following that order, from bottom to top, and modify labels assigned to upward edges. An important condition we will maintain, is that every vertex in an odd part $V_{2n+1}$ ($n \geq 1$) has all its downward edges (if any) labelled~$3$ or~$1$, while every vertex in an even part $V_{2n}$ ($n \geq 2$) has all its downward edges (if any) labelled~$2$ or~$1$. Note that this is trivially verified for the vertices in $V_k$, since they have no downward edges.

Assume we are currently considering a vertex $v$ in, say, an even part $V_{2n}$ with $n \geq 2$. By the hypothesis above, all downward edges of $v$ are labelled~$2$ or~$1$. Since all upward edges of $v$ are currently labelled~$1$, the $3$-degree of $v$ is currently $0$. Let us consider each of the $n$ parts $V_2,V_3,V_5,\dots,V_{2n-1}$. Recall that $v$ has a neighbour $u_i$ in each of these parts. Then we modify the label of each such edge $vu_i$ so that it becomes~$3$. This way, note that the $3$-degree of $v$ becomes exactly $n$, as required. Note also that we do not spoil the desired downward condition for the $u_i$'s. Now, depending on how many downward edges of $v$ are labelled~$2$, we claim that we can always turn to~$2$ the label of one or two upward edges so that $v$ gets bichromatic with odd $\{2,3\}$-degree as desired. Indeed, if $n \geq 3$, then we can freely change to~$2$ the label of an upward edge of $v$ to each of $V_4$ and $V_1$ to get $v$ as desired. If $n=2$, then note that $v$ might be missing at most one incident edge labelled~$2$. Indeed, if, on the one hand, an odd number of downward edges of $v$ are labelled~$2$, then $v$ is already bichromatic with odd $\{2,3\}$-degree. On the other hand, if an even number of downward edges of $v$ are labelled~$2$, then $v$ currently has even $\{2,3\}$-degree, in which case we make it odd by changing to~$2$ the label of an upward edge to $V_1$. This way, note that $v$ has to become bichromatic.

Similar arguments hold in the case when $v$ lies in an odd part $V_{2n+1}$ with $n \geq 1$. Again, all downward edges of $v$ are labelled~$3$ or~$1$, while all upward edges are currently labelled~$1$. We change to~$2$ the label of an upward edge of $v$ to each of the $n$ parts $V_1,V_4,V_6,\dots,V_{2n}$ so that $v$ has $2$-degree $n$. Now, we can make sure that $v$ is bichromatic with even $\{2,3\}$-degree in the following way. If $n \geq 2$, then we can change to~$3$ the label of an upward edge to $V_{2n}$ and/or $V_2$, if needed. If $n=1$, then note that $v$ currently has $2$-degree~$1$. If an odd number of downward edges are labelled~$3$, then $v$ is already bichromatic with even $\{2,3\}$-degree. If an even number of downward edges are labelled~$3$, then we can change to~$3$ the label of an upward edge to $V_2$ to achieve the same conclusion.

\medskip

By arguments above, only adjacent $1$-monochromatic vertices in $V_1$ and $V_2$ can be in conflict. More precisely, recall that the vertices of $V_1$ are $1$-monochromatic or $2$-monochromatic, while the vertices of $V_2$ are $1$-monochromatic or $3$-monochromatic. Another important property of the vertices in $V_3,\dots,V_k$ is that none of them has both $3$-degree~$1$ and odd $\{2,3\}$-degree at least~$3$ (bichromatic). In particular, assuming that, later on, we only turn vertices in $V_1$ into this \textit{special} type,  no conflict can involve special vertices.

As in the proof of Theorem~\ref{theorem:4chromatic}, let us define $\mathcal H$ as the set of components induced by the upward edges (all of which are currently labelled~$1$) of the $1$-monochromatic vertices of $V_2$. If $\mathcal H$ has no component on more than two vertices, then we are done. So let us focus on $H \in \mathcal H$, a component with order at least~$3$. Here as well, no vertex of $H$ is adjacent to a vertex in another component of $\mathcal H$, so we can again freely deal with $H$ without minding the other components. If no two $1$-monochromatic vertices in $H$ are adjacent, then we are done with $H$. So let us assume some adjacent vertices of $H$ are $1$-monochromatic, and some of these $1$-monochromatic vertices actually have at least two $1$-monochromatic neighbours (as otherwise we would be done as well).

We start by performing the following process:

\begin{enumerate}
    \item As long as $H$ has a $1$-monochromatic vertex $v \in V(H) \cap V_2$ with at least two $1$-monochromatic neighbours $u_1,u_2 \in V(H) \cap V_1$, we do the following:
    \begin{itemize}
        \item If $v$ has a $2$-monochromatic neighbour $u'$ in $V_1$ with $d_2(u')=2$, then we set $\ell(u'v)=3$.
        
        \item Otherwise, we set $\ell(vu_1)=\ell(vu_2)=2$.
    \end{itemize}
\end{enumerate}

By this process, every considered vertex $v$ of $V_2$ becomes either  $3$-monochromatic (first case), or $2$-monochromatic with $2$-degree~$2$ (second case). In the first case, all neighbours of $v$ in $V_1$ are special, $1$-monochromatic or $2$-monochromatic, thus not in conflict with $v$. In the second case, the neighbours of $v$ in $V_1$ are all special, $1$-monochromatic or $2$-monochromatic with $2$-degree different from~$2$, thus not in conflict with $v$.

We go on with the following process:

\begin{enumerate}
    \setcounter{enumi}{1}
    \item As long as $H$ has a $1$-monochromatic vertex $u \in V(H) \cap V_1$ with at least two $1$-monochromatic neighbours $v_1,v_2 \in V(H) \cap V_2$, we do the following:
    \begin{itemize}
        \item If $u$ does not have a $3$-monochromatic neighbour in $V_2$, then we set $\ell(uv_1)=\ell(uv_2)=3$.
        
        \item Otherwise, we do nothing.
    \end{itemize}
\end{enumerate}

Note that by repeatedly applying the first of these steps, no new conflict arises. This is because all $3$-monochromatic vertices we create in $V(H) \cap V_2$ have $3$-degree~$1$, while all $3$-monochromatic vertices we create in $V(H) \cap V_1$ have $3$-degree~$2$ while all their $3$-monochromatic neighbours have $3$-degree~$1$.

Let us now have a look at the subgraph $\mathcal B$ of $H$ induced by its remaining $1$-monochromatic vertices. Let us more particularly focus on the components of $\mathcal{B}$. If no such component has order greater than~$2$, then we are done. So let us focus on one component $B$ with order at least~$3$. Since previous Steps~1 and~2 have been performed as long as possible, $B$ must be a star with center $u \in V_1$ and at least two leaves $v_1,v_2$ in $V_2$, and $u$ has $3$-monochromatic neighbours in $V_2$.

Consider now the subgraph $\mathcal C$ of $G$ obtained from the vertices in $\mathcal B$ by adding the incident edges to their $3$-monochromatic neighbours in $V_2$. Note that this graph might have several components; let us focus on one of these components, say $C$. By construction, all vertices in $V(C) \cap V_1$ are $1$-monochromatic. Also, $C$ contains a $1$-monochromatic vertex $u \in V_1$ with two $1$-monochromatic neighbours $v_1,v_2 \in V_2$. Note furthermore that $v_1$ and $v_2$ have degree~$1$ in $C$ (as otherwise Step~1 above could have been applied once more). We now modify the labelling using the following analogue of Claim~\ref{claim:bipartite} (we omit a proof, as it would go along the exact same lines):

\begin{claim}\label{claim:bipartite2}
For every vertex $v$ in any part $V_i \in \{V_1,V_2\}$ of $C$, we can relabel the edges of $C$ with $1$ and $3$ so that $d_3(u)$ is odd for every $u \in V_i \setminus \{v\}$, and $d_3(u)$ is even for every $u \in V_{3-i}$. Similarly, we can relabel the edges of $C$ with $1$ and $3$ so that $d_3(u)$ is even for every $u \in V_i \setminus \{v\}$, and $d_3(u)$ is odd for every $u \in V_{3-i}$.
\end{claim}

We now use Claim~\ref{claim:bipartite2} as follows:

\begin{itemize}
    \item If $|V(C) \cap V_2|$ is even, then, by Claim~\ref{claim:bipartite2}, we can relabel the edges of $C$ with~$1$ and~$3$ so that all vertices in $V(C) \cap V_1$ have even $3$-degree, while all vertices in $V(C) \cap V_2$ have odd $3$-degree.
    
    \item If $|V(C) \cap V_2|$ is odd, then the same conclusion can be achieved in the subgraph $C-v_2$, since $|V(C-v_2) \cap V_2|$ is even.
\end{itemize}

By this modification, note that any two adjacent $3$-monochromatic vertices of $C$ have their $3$-degrees being of distinct parity, and are thus not in conflict. Recall in particular that, in earlier Step~2 above, all $3$-monochromatic vertices created in $V_1$ have even $3$-degree exactly~$2$ while all $3$-monochromatic vertices created in $V_2$ have odd $3$-degree exactly~$1$. Similarly, in earlier Step~1 above, every created $3$-monochromatic vertex was created in $V_2$ and has odd $3$-degree exactly~$1$. Thus, no conflict can involve $3$-monochromatic vertices.

Finally, we note that, in the process above, the only possible remaining conflict is actually in the second case, along the edge $uv_2$ since $v_2$ remains $1$-monochromatic while $u$ might have $3$-degree~$0$, and thus be $1$-monochromatic as well.
\end{proof}

To finish off this section, let us mention that the labelling scheme developed in the proof of Theorem~\ref{theorem:star} has another implication for the product version of the \textbf{1-2 Conjecture}, raised by Przyby{\l}o and Wo\'zniak in~\cite{PW10}. The 1-2 Conjecture asks whether every graph has an s-proper \textit{$2$-total-labelling} $\ell$, i.e., a $2$-total-labelling (assigning labels to edges and vertices) so that $\sigma_\ell^t(u) \neq \sigma_\ell^t(v)$ for every edge $uv$, where $\sigma_\ell^t(w)=\sigma_\ell(w)+\ell(w)$ for every vertex $w$. In other words, in this type of labelling we are also allowed to locally alter vertex sums (via vertex labels) without spoiling neighbouring ones.

In~\cite{SK12}, Skowronek-Kazi\'ow also introduced and studied the product version of the 1-2 Conjecture. By adapting existing proofs for the sum version of the 1-2 Conjecture, she mainly proved that every graph admits a p-proper total-labelling assigning labels in $\{1,2,3\}$ to the edges and labels in $\{1,2\}$ to the vertices. By modifying the last stage of our labelling scheme in the proof of Theorem~\ref{theorem:star}, we get another proof of that result.

\begin{theorem}
Every graph $G$ admits a p-proper total-labelling assigning labels in $\{1,2,3\}$ to the edges and labels in $\{1,2\}$ to the vertices.
\end{theorem}

\begin{proof}
Mimic the proof of Theorem~\ref{theorem:star}, until the last stage, i.e., to the point where vertices in $V_1$ are $1$-monochromatic or $2$-monochromatic, while the vertices in $V_2$ are $1$-monochromatic or $3$-monochromatic. Assign label~$1$ to all vertices, so that the products are not altered. To get a total-labelling as desired, we get rid of all remaining conflicts by just making sure that all vertices of $V_1$ become $2$-monochromatic. To that end, we simply change to~$2$ the label of every vertex in $V_1$.
\end{proof}

\section{A conjecture for $2$-labellings with restricted product conflicts} \label{section:weak-conjecture}

From a more general perspective, according to Theorem~\ref{theorem:star}, for every graph $G$ we can design a $3$-labelling such that $G[S_x]$ is a forest for every $x \geq 1$. One can naturally wonder whether $2$-labellings are powerful enough to achieve the same goal. As we did not manage to come up with any obvious reason why this could be wrong, we raise:

\begin{conjecture}\label{conjecture:weak}
Every graph $G$ can be $2$-labelled so that $G[S_x]$ is a forest for every $x \geq 1$.
\end{conjecture}

It is worth noting that Conjecture~\ref{conjecture:weak} matches a similar conjecture raised in~\cite{GWW15} by Gao, Wang and Wu in the sum context. They notably proved that the sum version of Conjecture~\ref{conjecture:weak} holds for graphs with maximum average degree at most~$3$ and series-parallel graphs. In what follows, as support, we prove Conjecture~\ref{conjecture:weak} (sometimes in an actually stronger form) for three classes of graphs: complete graphs, bipartite graphs, and subcubic graphs.

\begin{theorem}
Every complete graph $K_n$ admits a $2$-labelling such that one of the $S_x$'s induces an edge,
while all other $S_x$'s are independent sets.
\end{theorem}

\begin{proof}
We give an iterative labelling scheme which, starting from $K_2$, yields a desired $2$-labelling for larger and larger complete graphs $K_n$. To that end, we need a stronger hypothesis, namely that for every complete graph $K_n$ there is a desired $2$-labelling with the additional requirement that either there is no vertex incident only to edges labelled~$1$, or there is no vertex incident only to edges labelled~$2$.

This is true for $K_2$: by assigning label~$1$ to the only edge, we get a $2$-labelling where $S_1$ induces an edge (while there are no other $S_x$'s) and no vertex is incident only to edges labelled~$2$. Assume now our stronger claim is true for $K_{n-1}$ for some $n \geq 3$, and consider a $2$-labelling of $K_{n-1}$, with vertex set $\{v_1,\dots,v_{n-1}\}$, obtained by induction (thus with the desired properties). Let us extend this labelling to the incident edges of a newly-added vertex $v_n$ joined to all vertices in $\{v_1,\dots,v_{n-1}\}$, by assigning label~$1$ to all edges incident to $v_n$ if no vertex in $\{v_1,\dots,v_{n-1}\}$ is incident only to edges labelled~$1$, or by assigning label~$2$ to all edges incident to $v_n$ if no vertex in $\{v_1,\dots,v_{n-1}\}$ is incident only to edges labelled~$2$. Note that the $2$-degree of all vertices in $\{v_1,\dots,v_{n-1}\}$ grows by the same amount, either~$0$ or~$1$. Thus, no new conflict involving two vertices in $\{v_1,\dots,v_{n-1}\}$ arises. Now, regarding $v_n$, its $2$-degree is either the smallest possible ($0$) or the largest possible ($n-1$) for a vertex with degree $n-1$. By our choice of making $v_n$ incident to either only edges assigned label~$1$ or only edges assigned label~$2$, we deduce that $v_n$ cannot be involved in a conflict. Thus, there remains only one conflict, and there is either no vertex in $\{v_1,\dots,v_n\}$ incident only to edges labelled~$2$, or no vertex in $\{v_1,\dots,v_n\}$ incident only to edges labelled~$1$. This concludes the proof.
\end{proof}

\begin{theorem}
Every connected bipartite graph $G$ admits a $2$-labelling such that one of the $S_x$'s induces at most one star and isolated vertices,
while all other $S_x$'s are independent sets.
\end{theorem}

\begin{proof}
Let $v^*$ be any vertex of $G$. From $v^*$, we get a partition $V_0 \cup \dots \cup V_d$ of $V(G)$ where each $V_i$ contains the vertices at distance $i$ from $v^*$. Note that $V_0=\{v^*\}$. Since $G$ is bipartite, none of the $V_i$'s contains an edge. Furthermore, for every edge $uv$ we have $u \in V_i$ and $v \in V_{i+1}$ for some $i$. A part $V_i$ is said \textit{even} if $i$ is even, while $V_i$ is said \textit{odd} otherwise.

We produce a $2$-labelling $\ell$ of $G$ where every vertex in an even $V_i$ different from $V_0$ has even $2$-degree, while every vertex in an odd $V_i$ has odd $2$-degree. Note that the existence of  $\ell$ proves the claim, since, by such a labelling, $v^*$ is the only vertex from an even $V_i$ that can be involved in conflicts. In particular, if one of the $S_x$'s induces a graph containing a star, then that star must be centered at $v^*$.

We consider the vertices of $G$ different from $v^*$ successively, starting from those in the deepest $V_i$'s, and, as going up, finishing with those in $V_1$. In the course of this process, let us consider $v \in V_i$, a vertex in some $V_i$ all of whose incident edges going to $V_{i+1}$ (if any) have been labelled. By definition of the $V_i$'s, there is at least one edge incident to $v$ going to $V_{i-1}$. We assign label~$2$ to every such edge going to $V_{i-1}$, but maybe to one of them (to which we instead assign label~$1$) so that the $2$-degree of $v$ is of the desired parity.

Once all vertices different from $v^*$ have been treated that way, we end up with $\ell$ having the desired properties.
\end{proof}

\begin{theorem}
Every subcubic graph $G$ admits a $2$-labelling such that all $S_x$'s induce a forest.
\end{theorem}

\begin{proof}
We prove the claim by induction on $|V(G)|+|E(G)|$. As the claim can easily be proved when $G$ is small, we focus on proving the general case, which we do by considering the possible cases for the minimum degree $\delta(G)$ of $G$.

\begin{itemize}
    \item First assume $\delta(G)=1$, and let $v$ be a degree-$1$ vertex of $G$. Let us consider $G'=G-v$. By the induction hypothesis, there is a $2$-labelling of $G'$ which is as desired. We extend this labelling to $G$ by assigning label~$1$ to the edge incident to $v$. This way, note that the resulting labelling is as desired, since $G'[S_1]$ gets added a pending or isolated vertex.
    
    \item Next assume $\delta(G)=2$, and let $u$ be a degree-$2$ vertex with neighbours $v$ and $w$ of degree at least~$2$. We here consider $G'=G-u$, which has a $2$-labelling with the desired properties. Let us first try to extend this labelling to $G$ by assigning label~$1$ to $uv$ and $uw$. Note that if the desired properties are not met, then it must be because $G'[S_1]$ has a path from $v$ to $w$. In particular, both $v$ and $w$ have product~$1$, and each of these two vertices is adjacent, in $G'[S_1]$, to another vertex. Then, assign label~$2$ to $uv$ and label~$1$ to $uw$. Now the resulting labelling of $G$ must be as desired, since this removed $v$ from $G'[S_1]$, and added a pending or isolated path of length~$2$ to $G'[S_2]$. This is because $G$ is subcubic, which, at this point, implies that $v$ has at most one neighbour in $S_2$.
    
    \item Lastly assume $\delta(G)=3$, i.e., $G$ is cubic, and consider $u$ a degree-$3$ vertex with neighbours $v,w,x$ of degree~$3$. We consider $G'=G-u$, which, again, has a $2$-labelling with the desired properties. If we do not obtain a desired labelling of $G$ when assigning label~$1$ to $uv$, $uw$ and $ux$, then it must be because, say, $v$ and $w$ have product~$1$ and are joined by a path in $G'[S_1]$. By arguments above, due to the bounded maximum degree of $G$, if we do not obtain a desired labelling when assigning label~$2$ to $uv$ and label~$1$ to $uw$ and $ux$, then this must be because $x$ has product~$2$, and $G'[S_2]$ contains a path from $x$ to a neighbour of $v$. Then we deduce that, by the labelling of $G'$, in $G'$ the two remaining neighbours of $v$ have product~$1$ and~$2$, and $x$ has a neighbour with product~$2$. Then note that we are done when assigning label~$2$ to $uv$ and $ux$, and label~$1$ to $uw$. Indeed, this removes $v$ from $G'[S_1]$ and $x$ from $G'[S_2]$, adds to $G'[S_2]$ a pending or isolated edge (attached to $v$), and adds to $G'[S_4]$ a pending  or isolated path of length~$2$ (attached to $u$).
\end{itemize}

This concludes the proof.
\end{proof}


\end{document}